\documentclass[journal]{IEEEtran}
\usepackage{multicol}
\usepackage{cite}
\usepackage[dvips]{graphicx}
\usepackage{epsfig}
\graphicspath{{images/}}
\usepackage{subfigure}
\usepackage{float} 
\usepackage{gensymb}
\usepackage{array}
\usepackage{amsmath}
\usepackage{amsthm}
\usepackage{algorithm}
\usepackage{algorithmic}
\usepackage{mdwmath}
\usepackage{mdwtab}
\usepackage{amsfonts}
\usepackage{psfrag}
\usepackage{textcomp}
\newtheorem{proposition}{Proposition}

\newtheorem*{proof*}{Proof}

\usepackage{url}
\usepackage{color}
\begin{document}
\title{Analyzing the Trade-offs in Using Millimeter Wave Directional Links for High Data Rate Tactile Internet Applications}
\author{\IEEEauthorblockN{Kishor Chandra Joshi, Solmaz Niknam, R. Venkatesha Prasad and Balasubramaniam Natarajan, \textit{Senior Member, IEEE}}
\thanks{Kishor Chandra Joshi is with CNRS/CentraleSupelec, University of Paris-Saclay, Paris, France and TU Delft, Netherlands.}
\thanks{R. Venkatesha Prasad is with TU Delft, Netherlands.}
\thanks{Solmaz Niknam is with Virginia Tech.}
\thanks{Balasubramaniam Natarajan is with Kansas State University.}}
\maketitle
\begin{abstract}
Ultra-low latency and high reliability communications are the two defining characteristics of  Tactile Internet (TI). Nevertheless, some TI applications would also require high data-rate  transfer of audio-visual information to complement the haptic data. Using Millimeter wave (mmWave) communications  is an attractive choice for high datarate TI applications due to the  availability of large bandwidth in the mmWave bands. Moreover, mmWave radio access is also advantageous to attain the air-interface-diversity required for high reliability in TI systems as mmWave signal propagation significantly differs to sub-6\,GHz
 propagation. However, the use of narrow beamwidth in mmWave systems makes them susceptible to link misalignment-induced unreliability and high access latency.  In this paper, we analyze the trade-offs between high gain of narrow beamwidth antennas and corresponding susceptibility to misalignment in  mmWave links. To alleviate the effects of random antenna misalignment, we propose a beamwidth-adaptation scheme that significantly stabilize the link throughput performance.
\end{abstract}
\section{Introduction}
Tactile Internet (TI) aims to enable  interaction with remote environments in perceived real-time through the delivery of haptic information over wired/wireless networks. This has led to the two main requirements for network infrastructure facilitating TI, i.e., ultra-low latency and  high reliability~\cite{fettweis2014tactile}.
Although haptic information may be encoded in a few bytes, there are several applications such as robotic surgery, autonomous driving, where transfer of  high quality audio-visual information would be quintessential to complement the  haptic information~\cite{TI_IEEE}. This would add another requirement for TI, i.e., sustained high data-rate communication links. Consequently, such applications would require wireless interfaces that are highly-reliable, can provide ultra-low latency, and are also able to sustain high data rates in order of Multi-Gbps.  

 To achieve the carrier-grade reliability over wireless links, potential solutions include frequency diversity using multiple uncorrelated links and spatial diversity using simultaneous connectivity to multiple spatially-uncorrelated base stations\cite{aijaz2017realizing}. These approaches are also referred as interface-diversity techniques~\cite{interfaceDiversity}. To achieve the latency requirements at link-layer, short packet length is proposed to reduce the total transmission interval instead of transmitting long packets~\cite{massiveURRLC, MAC_UURLLC}. 

In the context of fifth-generation (5G) networks, millimeter Wave (mmWave) frequency band (30\,GHz to 300\,GHz) has emerged as a promising candidate for multi-Gbps wireless connectivity due to availability of large bandwidth chunks in mmWave frequency band~\cite{5G1,cogcel}. Since radio-interface diversity is important from reliability perspective, using mmWave radio-interfaces would certainly benefit the TI applications as mmWave signals propagation is highly likely to be uncorrelated to sub-6\,GHz signal propagation due to its differing propagation characteristics. In this context we envision a hybrid radio access architecture to support TI applications where sub-6\,GHz access will be used for transmitting haptic information while mmWave access will be used for high data-rate transmission of audio-visual information. However, using mmWave communications for TI application is not straight-forward. Specifically, the use of narrow beamwidth directional antennas, which is necessary to combat high-pathloss at mmWave frequencies~\cite{mmwavebeamforming}, can result in frequent link outages due to antenna misalignment~\cite{7158056}. Thus providing a sustained high data-rate radio interface for TI applications at mmWave frequencies is not a straight-forward task. \normalcolor

Ideally, a narrow beamwidth antenna should always result in better signal quality compared to a broader beamwidth antenna as antenna gain and beamwidth are reciprocal to each other. However, in practical situations, following  two factors have a significant impact on  the performance of a narrow beamwidth directional link:\newline
(i)~\textit{\textbf{Beam misalignment}} --  the alignment of Tx and Rx antennas can be disturbed  by many factors such as device holding pattern, random user movements, orientation error and  vibrations that would result in unstable link quality. Since link availability is one of the most important indicator of link reliability~\cite{availabilityAnalysis}, it is highly important to analyze the impact of beam alignment on link quality for different transmitter and receiver beamwidths. This is specifically interesting for TI applications where link outages are not at all desired.\normalcolor\newline
(ii)~\textit{\textbf{Beam setup time}}-- finding best Tx and Rx antenna directions is essential to establish a mmWave link. mmWave standards IEEE 802.15.3c~\cite{iee:IEEE802.15.3c}) and  (IEEE 802.11ad~\cite{IEEE802.11ad} have proposed beamforming-protocols to find the Tx and Rx beams that result in best signal quality.  The beam-search space is inversely proportional to the beamwidths of Tx and Rx antennas. In the case of narrow beamwidths, a significant fraction of allocated time-slot can be exhausted  finding the best Tx/Rx beams. Hence the beam-searching can be a considerable overhead in narrow-beamwidth mmWave links. 

Thus the use of directional antennas at mmWave frequencies impacts both the reliability and latency performance which are of paramount importance to TI applications. Further, the TI applications do involve movement, at least within a constrained domain, thus proper beam forming and alignment is a must. It is important to ensure that the random misalignment does not impact  the link performance; and when link disruption is unavoidable, beam setup procedure should be fast enough to minimize the impact of outage.

The available literature on  mmWave communications has mainly dealt with capacity analysis or beamforming design. There are very few papers that consider mmWave communications for ultra low latency and high reliability communication (URLLC) scenarios~\cite{Mezzavilla_latency,Ming_latency,Ming_latency2,debbah_latency,ShivendraPanwar}. In \cite{Mezzavilla_latency}, various challenges of achieving URLLC at mmWave band are highlighted. It is argued that a significant reworking of the entire protocol stack (short frame size at physical layer, dynamic MAC protocols, moving content closer to edge, etc.) is required if mmWave radio access is used for URLLC. In \cite{Ming_latency}, it is shown that cooperative networking can significantly improve the latency performance of  mmWave based heterogeneous networks. In \cite{Ming_latency2}, two strategies, namely, traffic dispersion and networking densification are proposed to reduce the end to end in mmWave wireless networks. Here dispersion stands for offloading traffic to different spatial paths using distributed antenna systems while densification refers to increasing the density of mmWave BSs. It is shown that both the strategies improve  latency performance for a given sum budget power. In \cite{debbah_latency}, the problem
of reliability and latency in mmWave massive
multiple-input multiple-output (MIMO) networks is studied by using the Lyapunov technique that follows utility-delay control approach  and successfully adapts to channel variations and the queue dynamics. In \cite{ShivendraPanwar}, the feasibility of using mmWave access for URLLC considering dynamic blockages is considered. It is shown that the optimal BS deployment is driven by reliability and latency constraints instead of coverage and rate requirements.

 All the above papers follow a system level modeling approach and do not account for the impact of beamwidth on performance of individual links considering medium access control~(MAC) overheads amounting to use of directional antennas. In \cite{adaptive, 80211ad}, the performance evaluations of contention access in IEEE 802.11ad MAC protocol is done assuming coarse-sector beamwidths and without considering alignment and searching overheads.   There are few papers that separately consider the beam searching overhead~\cite{efficientbeamswitch,BBS_Nitsche,beamsearchingKTH} or the misalignment~\cite{Misalignment_Wildeman, Ergodic_RWHeath, yu2017coverage}. In \cite{efficientbeamswitch}, an efficient beam switching mechanism is proposed that utilizes a modified Rosenbrock numerical algorithm to select the best beam-pair. Using the direction estimates at 2.4 and 5\,GHz, it is possible to infer the coarse beam directions of 60\,GHz links~\cite{BBS_Nitsche}. Since direction estimates at 2.4/5\,GHz is obtained using passively-heard frames, it enormously reduces the searching overheads. In \cite{liu2017millimeter}, it is shown that the exhaustive search outperforms hierarchal search in terms of asymptotic misalignment probability. In \cite{beamsearchingKTH},  capacity of 60\,GHz link is analysed considering beam-searching overhead while using an ideal flat-top antenna model. This paper does not take into account the impact of antenna misalignment. Moreover, the flat-top antenna model assumed in this paper  transforms the continuously varying antenna gains into binary  values that are not suitable for analysing the impact of misalignment.  Our work is close to this paper, however we consider a Gaussian antenna model and  include misalignment errors as well into our link modeling framework. In \cite{Ergodic_RWHeath}, it is numerically demonstrated that the Gaussian main-lobe antenna pattern, which we have adopted in this paper, can represent the mmWave directional antennas with a reasonable accuracy.The  rate analysis of mmWave mobile network using  stochastic geometry is presented in \cite{Misalignment_Wildeman}. However,  beam set up overhead is ignored and mainly employ flat-top antenna models thus unable to provide accurate impact of misalignment. In \cite{yu2017coverage}, authors also use stochastic geometry framework and employ Sinc and Cosine antenna patterns at BSs to evaluate the impact of misalignment. However, beam-searching overhead is not considered. 

  In summary, despite the plethora of recent literature on mmWave communications, there is no prior work that has holistically considered the consequential trade-offs of  using narrow beamwidths, i.e., increase in antenna gains at the cost of highly un-sustained link quality that is bound to increase the susceptibility to link unreliability and high-latency. \normalcolor In this paper we aim to fill this gap. Our main contributions are as follows.
\begin{itemize}
\item {We develop a novel link capacity optimization framework that shows the existence of optimum TX/Rx beamwidths. Our model jointly considers the beam-misalignment, beam-searching overhead and allocated  time-slot length to model the capacity of mmWave links.}
\item {We show that although the narrow beams  result in high throughput links, the resulting links are highly susceptible to beam-misalignment. Such links are highly un-reliable and are not suitable to support the TI applications.}
\item { We propose a beamwidth-adaptation scheme that is misalignment-aware and  significantly stabilizes the quality of mmWave links and result in (50\% to 150\%) improvement in average link capacity.}
\end{itemize}  
 \begin{figure*}[!t]
    \centering
    \subfigure[Network architecture.]{ 
    \includegraphics[width=0.36\textwidth]{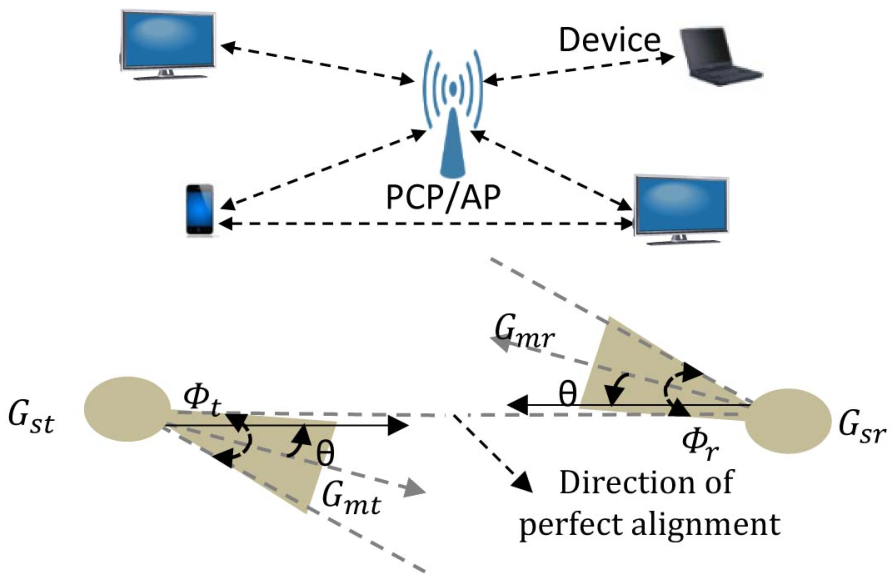}
    \label{fig_systemmodel_21}
    }
    \subfigure[Beacon Interval.]{
    \includegraphics[width=0.36\textwidth]{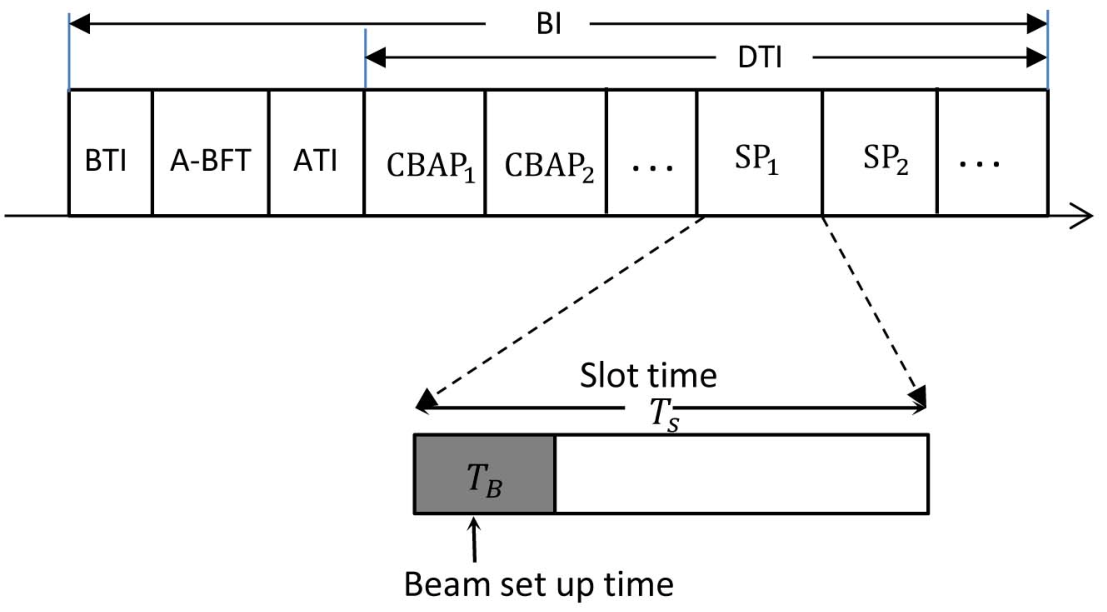}
    \label{fig:systemmodel_22}
    }
    \caption{IEEE 802.11ad network architecture and beacon interval.}
    \label{fig:systemmodel_2}
    \end{figure*}
\section{Preliminaries of IEEE 802.11ad}\label{sec:tradeoffs_descript}
Fig.~\ref{fig_systemmodel_21} shows an IEEE 802.11ad Personal Basic Service Set (PBSS) formed by the 60\,GHz stations (STAs) where one of them acts as PBSS Control Point/Access Point (PCP/AP). 
IEEE 802.11ad employs a hybrid medium access control protocol consisting of both contention-based channel access (CSMA/CA) and fixed TDMA based channel access.  Fig.~\ref{fig:systemmodel_22} shows the timings in an IEEE 802.11ad MAC that are based on beacon intervals~(BI) which consists of: (i)~beacon transmission interval~(BTI); (ii)~association beamforming training (A-BFT) where STAs associate with the PCP/AP; (iii)~announcement time interval~(ATI) where the exchange of management information between PCP/AP and the STAs happens; and (iv)~data transfer interval~(DTI), which consists of contention-based access periods (CBAPs) and fixed-access service periods~(SPs). CBAPs use CSMA/CA based channel access mechanism while during SPs, TDMA access mechanism is used. In this paper, we mainly focus on the SPs part of the BI where  high speed data transmission using narrow beamwidth happens.

\subsection{Antenna model}
To examine the effect of beam alignment errors on the communication link performance,  detailed mathematical models of antennas are needed. Generally, capturing all  essential characteristics of the real-world antennas in a simple analytical expressions/model is difficult. The most widely used cone-plus-circle antenna model~\cite{beamsearchingKTH}  assumes constant gains for both the main and the side lobes. The assumption of constant gain for main lobe  makes it unsuitable for examining the effect of misalignment. In this case, any alignment error  smaller than the  beamwidth of  main lobe would be unnoticed. We use a relatively pragmatic antenna model employing the Gaussian main lobe radiation pattern, which is proposed in IEEE 802.15.3c~\cite{iee:IEEE802.15.3c}. Let $G_m^{\phi}(\theta)$ and $	G_s^{\phi}$ represent the mainlobe and sidelobe gains, respectively. The analytical model is represented by,
\begin{small}
	\begin{align}\label{Eq:GaussianGain}
	G^{\phi}(\theta)=
	\begin{cases}
	\begin{array}{ll}
	G_m^{\phi}(\theta)=\left(\frac{1.6162}{\sin(\frac{\phi}{2})}\right)^2e^{-K_1{4\log_e(2)}\left(\frac{\theta}{\phi}\right)^2},&|\theta|\leq 1.3\phi,  \\
	G_s^{\phi}=e^{-2.437}\phi^{-0.094},&|\theta|>1.3\phi.
	\end{array}
	\end{cases}
	\end{align}
\end{small}
Where, $\phi$ is the half power beamwidth (HPBW) beamwidth and $\theta$ is the deviation angle of antenna boresight from axis of perfect alignment.
The main lobe beamwidth $\phi_{ML}$ (defined as frequency intervals between -20\,dB gain levels relative to the peak gain) is approximately  $\phi_{ML}=2.6\phi$. 

\begin{figure}[h]
	\centering
		\psfrag{x}[cc][cc][0.75]{\small{Misalignment, $\theta$\textdegree}}
	\psfrag{y}[cc][cc][0.75]{\small{Gain (dB)}}
	\includegraphics[width=0.27\textwidth]{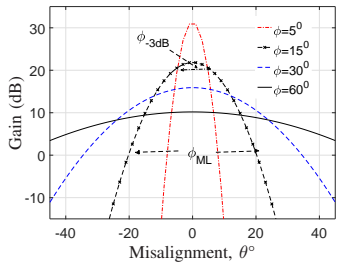}
	\caption{Gain vs beam misalignment.}
	\label{fig_GainMisalignment}
\end{figure}
The  the main lobe gain  with varying  alignment error for different HPBW $\phi$ is shown in Fig.~\ref{fig_GainMisalignment}. We  can observe that the gain curve  of smaller beamwidth decays faster than the wider beamwidth. This implies that the  narrow beam (high gain) antennas are most affected by the alignment errors. Consequently, in presence of alignment errors, the effective gains of Tx and Rx  would be way less than the expected gains. Using the Gaussian main lobe offers different gains for different value of alignment errors as opposed to the ideal cone-plus-circle antenna model where antenna gain has binary values, thus making it possible to track the impact of alignment errors on the link quality.  

It should be noted that if  maximum possible misalignment, denoted as $|\theta_{max}|$, is greater than the half of main lobe beamwidth, antenna pointing error of main lobe would lead to link failure as antennas would no longer be able to communicate using the main lobe. Therefore, to ensure the link availability, \begin{small}$\label{LowerLimitBeamwidth}\frac{\phi_{ML}}{2}\ge |\theta_{max}|$\end{small} must be satisfied.
\subsection{Beam Training procedure}\label{sec:beamtraining_overhead}  IEEE  IEEE 802.11ad  use a 2-level beamforming protocol.  In the $1^{st}$ stage, devices pair that want to establish a connection start scanning, using wider beamwidths called quasi-omni (QO) levels or sector levels. Generally, the beamwidth of sector-level could be 180$\degree$ or 90$\degree$. Once the best Tx and Rx sectors are found,  high resolution beams (1\textdegree-10\textdegree) are used in the $2^{nd}$ stage of beamforming procedure.

 Fig.~\ref{fig_antpattarens} depicts  the sector and beam level beamwidths where  many fine beams are contained within a coarse sector. To identify the best Tx and Rx beams, training sequences are transmitted in all the possible directions.  Let  the sector level Tx and Rx beamwidths are denoted by $\Omega_t$ and $\Omega_r$, respectively. Let  the beam level Tx and Rx beamwidth are denoted by $\phi_t$ and $\phi_r$, respectively.  Then, the total number of beam directions to be probed becomes  $\frac{\Omega_t}{\phi_t}$ Tx  and $\frac{\Omega_r}{\phi_x}$ Rx beams, respectively. Therefore, the total time required to search the best Tx and Rx beams (represented by $T_B$) can be given as,
\begin{small}\begin{equation}\label{Eq:timeslotAllocated}
T_B^{\phi_t, \phi_r}=\left(\frac{2\pi}{\Omega_t}+\frac{2\pi}{\Omega_r}\right)T_p+\left(\frac{\Omega_t}{\phi_t}+\frac{\Omega_r}{\phi_r}\right)T_p,
\end{equation}\end{small}
here, $T_p$ represents the transmission time of a beam training sequence.  According to IEEE 802.11ad, the $1^{st}$ stage beamforming is performed during A-BFT period while the $2^{nd}$ stage beamforming may be performed during SP or during A-BFT. To consider the effect of beam-training overhead on the effective  channel time available for data transmission, we assume that both stages of beamforming are performed during SP.\normalcolor

As depicted in Fig.~\ref{fig:systemmodel_22}, a  slot $T_S$ granted to a device pair is divided into  two parts: (i)~$T_B$ which is used for beam training; and (ii)~$T_S-T_B$ which is used for data transmission. This implies that the beam training time $T_B$ impacts the effective transmission time.  From (\ref{Eq:timeslotAllocated}), we deduce that the  beam training overhead $T_B$ is dependent on  Tx/Rx beamwidths assuming a fixed $T_p$. As the beamwidth decreases, the beam-search space expands thus forcing a trade-off between the high antenna gains and the corresponding increase in the beam searching overhead.
\begin{figure}[!]
\centering
\psfrag{x}[cc][cc][0.75]{\small{Beamwidth, $\phi$\textdegree}}
\psfrag{y}[cc][cc][0.75]{\small{Capacity (bits/slot/Hz)}}
\includegraphics[width=0.270\textwidth]{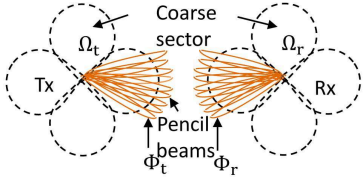}
\vspace{-10mm}
\caption{IEEE 802.11ad beam patterns.}
\label{fig_antpattarens}
\end{figure}
\section{Modeling of link capacity considering beam setup time and misalignment}\label{sec:JointDescription}
Let $P_t$ be the Tx power, $\lambda$ be the  carrier wavelength, $\alpha$ be the path loss exponent, and  $G_t$, $G_r$ be the Tx and Rx antenna gains, respectively. By applying the Friis's free-space pathloss equation, the received power $P_r$ at a  $d$\,m from the Tx is $P_r=G_tG_rG_0(\lambda,d)P_t$.  Where, $G_0(\lambda,d)=\left(\frac{\lambda}{4\pi d}\right)^\alpha$ represents the overall effect of propagation and free-space path losses if Tx and Rx are separated by $d$\,m. We define  $\eta_{\phi_t, \phi_r}$  as the fraction of allocated time slot $T_S$  used for data transmission, i.e.,  $\eta_{\phi_t, \phi_r}=\left( 1-\frac{T_{B}^{\phi_t, \phi_r}}{T_s}\right)$. \normalcolor Then,  the achievable data rate ($r^{\phi_t, \phi_r}$ (bits/slot/Hz)) for a given Tx-Rx beamwidths is given by,
\begin{align} \label{jointEquation1}
r^{\phi_t, \phi_r}(\theta_t, \theta_r)=\eta_{\phi_t, \phi_r}\log_2\left(1+\frac{G^{\phi_t}_t(\theta)G^{\phi_r}_r(\theta)G_0(\lambda,d)P_t}{N_0W}\right).
\end{align}\normalsize
Here, $N_0$ is the  white Gaussian noise's one-sided power spectral density, $W$ represents the signal bandwidth and Tx and Rx antenna gains $G^{\phi_t}_t(\theta)$  and $G^{\phi_r}_r(\theta)$ are given in (\ref{Eq:GaussianGain}). $\eta_{\phi_t, \phi_r}$, as previously defined,  depends on sector beamwidths $\Omega_t$ and $\Omega_r$, pencil beam beamwidths $\phi_t$ and $\phi_r$, and allocated time slot $T_S$. 
%
The objective is to find the optimum Tx and Rx beamwidths ($\phi_t$ and $\phi_r$) that can maximize the link capacity $r^{\phi_t, \phi_r}(\theta_t, \theta_r)$. Here, for simplicity and without loss of generality we assume $\theta_r=\theta_t=\theta$. Therefore, the resulted  optimization problem can be written as,
\begin{align}\label{opt1} \notag
&\mathbb{P}1: \quad \mathop {\max }\limits_{{\phi _r},{\phi _t}} r({\phi _r},{\phi _t},{\theta}) \\
&\hspace{1.5cm}= \eta_ {{\phi _r},{\phi _t}}{\rm{log}}\left(1 + {c_1({\phi _r},{\phi _t})}{{\rm e}^{ - {{(\frac{{{\theta}}}{{{\phi _r}}})}^2}- {{(\frac{{{\theta}}}{{{\phi _t}}})}^2}  }}\right)\\\nonumber
&\hspace{1.1cm} s.t.\,\, |\theta|\leq \min [2.6\phi_{t}, 2.6\phi_{r}]\\\nonumber
&\hspace{1.5cm}\,\,\,\,  \phi_t\leq \Omega_t, \phi_r\leq \Omega_r
\end{align}\normalcolor
Here, ${\Omega _r},{\Omega _t} \in \left( {0,2\pi} \right]$ and  ${c_1}({\phi _r},{\phi _t})$ (assuming both Tx and Rx stay within their main lobes) is,
\begin{align}\label{eta_c}
c_1({\phi _r},{\phi _t})={{{(\frac{\lambda }{{4\pi d}})}^\alpha}\frac{{{P_t}}}{{{N_0}W}}\left( {\frac{{{{1.6161}^4}}}{{{{\sin }^2}(\frac{{{\phi _r}}}{2}){{\sin }^2}(\frac{{{\phi _t}}}{2})}}} \right)}
\end{align}
 Since, we do not have any prior information on random variable $\theta$, we assume $\theta$ is uniformly distributed in ${\theta}\in \left[ {-\theta_m, \theta_m} \right]$ with pdf $f(\theta)$ given as,

\begin{small}\begin{align}\label{Eq:pdf_theta}
f(\theta)=
\begin{cases}
\begin{array}{ll}
\frac{1}{2\theta_m},&|\theta|\leq \theta_m,  \\
0,&|\theta|> \theta_m.
\end{array}
\end{cases}
\end{align}\end{small}
 The optimization problem $\mathbb{P}1$ is a robust optimization problem in which a certain measure of robustness is sought with respect to the random variable $\theta$.
\newtheorem{lemma}{Lemma}
\begin{lemma}\label{lem:lemma1}
The objective function of the optimization problem $\mathbb{P}1$ is a concave function with respect to the variable $\theta$.
\end{lemma}
\begin{proof}
The objective function of the problem $\mathbb{P}1$ is in the form of ${\rm log}(1+{\rm e}^{-x^2})$. Since, $-x^2$ is concave function and ${\rm e}(.)$ and ${\rm log}(.)$ are functions that preserve concavity \cite{boyd2009convex}, the overall function is still a concave function.
\end{proof}
Given \emph{Lemma} \ref{lem:lemma1}, we invoke  Jensen's inequality. Therefore,

\begin{align} \label{} \notag
&E_\theta\left[ {\eta_ {{\phi _r},{\phi _t}}{\rm log}\big(1 + {c_1}({\phi _r},{\phi _t}){{\rm e}^{ - (\frac{{{\theta ^2}}}{{{\phi _r}^2}})- (\frac{{{\theta ^2}}}{{{\phi _t}^2}})  }}\big)} \right]\\ \notag
&\hspace{1cm}\le \eta_ {{\phi _r},{\phi _t}}{\rm log}\big(1 + {c_1}({\phi _r},{\phi _t}){{\rm e}^{ - (\frac{{E[{\theta ^2}]}}{{{\phi _r}^2}})- (\frac{{E[{\theta ^2}]}}{{{\phi _t}^2}}) }}\big)\\
&\hspace{1cm}{\mathop  = \limits^{\left( a \right)}} \eta_ {{\phi _r},{\phi _t}}{\rm log}\big(1 + {c_1}({\phi _r},{\phi _t}){{\rm e}^{ - (\frac{{\frac{\theta_m^2}{3}}}{{{\phi _r}^2}})- (\frac{{\frac{\theta_m^2}{3}}}{{{\phi _t}^2}})  }}\big)
\end{align}
where $(a)$ follows from the assumption of $\theta$ as a zero-mean uniformly distributed random variable. Therefore, the optimization problem can be rewritten as
\begin{small}\begin{align}\label{opt2} \notag
&\mathbb{P}2: \quad \mathop {\max }\limits_{{\phi _r},{\phi _t}} r^{{\phi _r},{\phi _t}} \\
&\hspace{1.5cm}= \eta_ {{\phi _r},{\phi _t}}{\rm{log}}\bigg(1 + {{{{(\frac{\lambda }{{4\pi d}})}^\alpha}\frac{{{P_t}}}{{{N_0}W}}
	\bigg( {\frac{{{{1.6161}^4}}}{{{{\sin }^2}(\frac{{{\phi _r}}}{2}){{\sin }^2}(\frac{{{\phi _t}}}{2})}}} \bigg)}}\nonumber \\
&\hspace{1.5cm}\times{{\rm e}^{ - {{(\frac{{{\theta_m}}}{{{3 \phi _r}}})}^2}- {{(\frac{{{\theta_m}}}{{{3 \phi _r}}})}^2}}}\bigg)
\end{align}\end{small}
\begin{lemma}\label{lem:lemma2}
The objective function of the optimization problem $\mathbb{P}2$ is a strictly quasi-concave function of variables $\phi _r$ and $\phi _t$.
\end{lemma}
\begin{proof}
In order to prove \emph{Lemma} \ref{lem:lemma2}, considering the fact that the  objective function is a symmetric function in $\phi_r$ and $\phi_t$, we first prove that a function of the form $\frac{{{{\rm{e}}^{ - {x^{ - 2}}}}}}{{{{\sin }^2}\left( {\frac{x}{2}} \right)}}$ is a strictly quasi-concave function.

Considering the theorem\footnote{$f$ is strictly quasi-concave \emph{iff} $\forall a {\in} \Re ,\,\forall \lambda{\in} \left[ {0,1} \right], \text{and} \,\forall x,y$\\
\begin{align} \label{theo:strict_conacve}
f(x) > a,\,\,f(y) > a \Rightarrow \,\,f(\lambda x + (1 - \lambda )y) > a.
\end{align}} from \cite{boyd2009convex}, if 
$\frac{{{{\rm{e}}^{ - {x^{ - 2}}}}}}{{{{\sin }^2}\left( {\frac{x}{2}} \right)}}> a$ and $\frac{{{{\rm{e}}^{ - {y^{ - 2}}}}}}{{{{\sin }^2}\left( {\frac{y}{2}} \right)}}> a$, then,
\begin{align} \label{} \notag
&f(\lambda x + (1 - \lambda )y) \\ \notag
&\hspace{1cm} = \frac{{{{\rm{e}}^{ - {{\left( {\lambda x + (1 - \lambda )y} \right)}^{ - 2}}}}}}{{{{\sin }^2}\left( {\frac{{\left( {\lambda x + (1 - \lambda )y} \right)}}{2}} \right)}} > \left\{ \begin{array}{l}
\frac{{{{\rm{e}}^{ - {y^{ - 2}}}}}}{{{{\sin }^2}\left( {\frac{x}{2}} \right)}}\,\,\,\,\,\,\,\,{\rm{if}}\,\,\,x \ge y\\ \notag
\frac{{{{\rm{e}}^{ - {x^{ - 2}}}}}}{{{{\sin }^2}\left( {\frac{y}{2}} \right)}}\,\,\,\,\,\,\,\,{\rm{if}}\,\,\,x < y
\end{array} \right.\\
&\hspace{1cm} > \left\{ \begin{array}{l}
\frac{{a{{\sin }^2}\left( {\frac{y}{2}} \right)}}{{{{\sin }^2}\left( {\frac{x}{2}} \right)}}\,\,\,\,\,\,\,\,{\rm{if}}\,\,\,x \ge y\\
\frac{{a{{\sin }^2}\left( {\frac{x}{2}} \right)}}{{{{\sin }^2}\left( {\frac{y}{2}} \right)}}\,\,\,\,\,\,\,\,{\rm{if}}\,\,\,x < y
\end{array} >a. \right.
\end{align}
Therefore, functions of the form $\frac{{{{\rm{e}}^{ - {x^{ - 2}}}}}}{{{{\sin }^2}\left( {\frac{x}{2}} \right)}}$, are quasi-concave.
In addition, if $f (.)$ is a strictly quasi-concave function and $g$ is an increasing function, the composite function $g(f(.))$ is a strictly quasi-concave function \cite{boyd2009convex}. Therefore, given that ${\rm log}(1+x)$ is an increasing function of $x$, ${\rm log}(1+{\frac{{{{\rm{e}}^{ - {x^{ - 2}}}}}}{{{{\sin }^2}\left( {\frac{x}{2}} \right)}}})$ is strictly quasi-concave. 
Moreover, ${1 - (\frac{{2\pi }}{{{\Omega _t}}} + \frac{{2\pi }}{{{\Omega _r}}})\frac{{{T_p}}}{{{T_s}}} - (\frac{{{\Omega _t}}}{{{x}}})\frac{{{T_p}}}{{{T_s}}}}$ is a nonnegative strictly quasi-concave function, given the typical values of its parameters $T_p$, $T_s$, $\Omega_t$ and $\Omega_r$ (Proof comes in the Appendix). Given the fact that the product of nonnegative strictly quasi-concave functions is quasi-concave~\cite{kopa2012characterization}, the overall function which is in the form of $\left({1 - (\frac{{2\pi }}{{{\Omega _t}}} + \frac{{2\pi }}{{{\Omega _r}}})\frac{{{T_p}}}{{{T_s}}} - (\frac{{{\Omega _t}}}{{{x}}})\frac{{{T_p}}}{{{T_s}}}}\right){\rm log}(1+{\frac{{{{\rm{e}}^{ - {x^{ - 2}}}}}}{{{{\sin }^2}\left( {\frac{x}{2}} \right)}}})$ is also a strictly quasi-concave function.
\end{proof}

Since, the objective function in $\mathbb{P}2$ is strictly quasi-concave, the maximizer is unique \cite{boyd2009convex}. Given that the objective function in \eqref{opt2} is a differentiable function with respect to $\phi _r$ and $\phi _t$, the maximum value is achieved by the well-known necessary and sufficient condition $\nabla r({\phi _r},{\phi _t}) = 0$. The derivative of the objective function with respect to $\phi _r$ and $\phi _t$ are given in \eqref{eq:der_phiR} and \eqref{eq:der_phiT}. Therefore, we have a set of two equations of two variables. However, finding the solution of the equation $\nabla r({\phi _r},{\phi _t}) = 0$ is analytically complex. Therefore, in such cases the problem must
be solved using  iterative algorithms~\cite{boyd2009convex}.  

\begin{align} \label{eq:der_phiR}\notag
\frac{\partial }{{\partial {\phi _r}}}&=\frac{\bigg( \eta_{\phi_t, \phi_r}c_1({\phi _t},{\phi _r}){\rm e}^{-\frac{\theta^2_m}{3\phi^2_r}}\left( \frac{2\theta^2_m}{3\phi^3_r}-\cot{(\frac{{{\phi _r}}}{2}}\right)\bigg)}{\left(1+c_1({\phi _t},{\phi _r}){\rm e}^{-\frac{\theta^2_m}{3\phi^2_r}}\right)\log2}\\
&+\frac{\Omega_r\,T_p}{T_s\,\phi_r^2\log2}=0
\end{align}
\begin{align} \label{eq:der_phiT}\notag
\frac{\partial }{{\partial {\phi _t}}}&=\frac{\bigg( \eta_{\phi_t, \phi_r}c_1({\phi _t},{\phi _r}){\rm e}^{-\frac{\theta^2_m}{3\phi^2_t}}\left( \frac{2\theta^2_m}{3\phi^3_t}-\cot{(\frac{{{\phi _t}}}{2}}\right)\bigg)}{\left(1+c_1({\phi _t},{\phi _r}){\rm e}^{-\frac{\theta^2_m}{3\phi^2_t}}\right)\log2}\\
&+\frac{\Omega_t\,T_p}{T_s\,\phi_t^2\log2}=0
\end{align}

Fig.~\ref{fig_KSU_params_corrected_45degree} shows the variations in capacity $r^{{\phi _r},{\phi _t}}$ as function of $\phi_t$ and $\phi_r$ for $\theta_m$=10\degree. We can observe that there exist optimum $\phi_t$ and $\phi_r$ that maximize the link capacity.
 \begin{figure}[]
 \centering
 \includegraphics[width=0.35\textwidth]{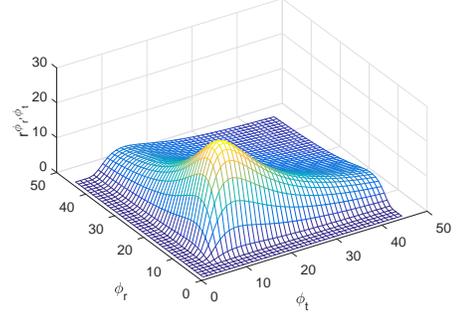}
  \caption{Link capacity for varying $\phi_t$, and $\phi_r$, $\theta_m=10\degree$.}
 \label{fig_KSU_params_corrected_45degree}
  \end{figure}
\section{Calculating Average Link Capacity}
Generally in most of the scenarios AP/BS are fixed and mainly users' devices are more likely to experience alignment errors. Therefore, for analytical tractability we assume that only user devices are affected by alignment errors. Hence,  $\theta_r=\theta$ and $\theta_t=0$. By substituting the expressions for antenna gain  in (\ref{jointEquation1}), 
\begin{small}
\begin{align} \label{jointEquation2}
r^{\phi_t, \phi_r}(\theta)=\eta_{\phi_t, \phi_r}\log_2\left(1+c_1\exp{\left(-\left(\frac{\theta}{\phi_r}\right)^2\right)}\right).
\end{align}
\end{small}
For brevity,  we  denote   $\eta_{\phi_t, \phi_r}$ and $r^{\phi_t, \phi_r}(\theta)$  by   $\eta$ and $r$, respectively.

Since  $\theta$ is assumed to be a uniformly distributed random variable interval $-\theta_m$ to $\theta_m$. Hence $r$,  a function of  $\theta$ ($r=g(\theta)$), also become a random variable. We represent the resulting random variable by $R$.  Applying the rule of transformation of random variables~\cite{grimmett2001probability}, the pdf of $R$ represented as $f_{R}(r)$, , can be written as,
\begin{equation}
f_{R}(r)=(-1)^i f(g^{-1}(r))\frac{\partial }{\partial r}g^{-1}(r).
\end{equation}
Where, $i = 0 $ if $g(\theta)$ is a monotonically increasing function of $\theta$, and $i=1$ if it is a monotonically decreasing function of $\theta$. It is evident from (\ref{jointEquation2}) that $r=g(\theta)$ is a monotonically decreasing function of $\theta$, which maximizes at $\theta=0$, and minimizes at $|\theta|=\theta_m$. Consequently, $i=1$ and $r$ is bounded between $g(\theta_m)$ and  $g(0)$. Using  (\ref{jointEquation2}), 
\begin{small}
\begin{align} \label{jointEquation2_inverse}
g^{-1}(r)={\phi_r}\left(\ln{\frac{2^{\frac{r}{\eta}}-1}{c_1}}\right)^{-\frac{1}{2}}.
\end{align}
\end{small}

From (\ref{Eq:pdf_theta}), it can be inferred that $f(g^{-1}(r))=\frac{1}{2\theta_{m}}$. Hence, the pdf of $R$ turns out as \begin{small}$f_{R}(r)=-\frac{\partial }{\partial r}g^{-1}(r)$\end{small}. Finally, the expression for average link capacity can be written as,
\small\begin{align}\label{expectedcapacity}
\mathbb{E}[R]=\int_{g(\theta_m)}^{g(0)}rf_{R}(r)dr.
\end{align}\normalsize
Since $f_{R}(r)=-\frac{\partial }{\partial r}g^{-1}(r)$, (\ref{expectedcapacity}) can be simplified  as follows,
\small\begin{align}\label{finalExpectedCapacity}
\mathbb{E}[R]=-\frac{1}{2\theta_{m}}\left(rg^{-1}(r)-\int{g^{-1}(r)}dr\right)\Bigg|_{g(\theta_m)}^{g(0)}.
\end{align}\normalsize
Since we are interested in examining the impact of alignment error on the main lobe gain, we assume that antenna pointing directions are within the -20\,dBm gain directions. Therefore, we can safely assume that ${2^{\frac{r}{\eta}}>> 1}$. Hence, the resulting approximate expression for 
 $g^{-1}(r)$ from (\ref{jointEquation2_inverse}) is,
\begin{align} \label{jointEquation2_inverse_approx}
\widetilde{g}^{-1}(r)={\phi_r}\left(\ln{\frac{2^{\frac{r}{\eta}}}{c_1}}\right)^{-\frac{1}{2}}.
\end{align}
The above approximation is required to find an integrable expression for $g^{-1}(r)$ as the exact expression in (\ref{jointEquation2_inverse}) is not integrable. 
Using (\ref{jointEquation2_inverse_approx}), (\ref{finalExpectedCapacity}) can be simplified as,
\small\begin{align}\label{finalExpectedCapacity2}
\mathbb{E}[R]=\frac{1}{2\theta_{m}}\left(\frac{2\phi_r\eta}{\ln{2}}\left(\ln{\frac{2^{\frac{r}{\eta}}}{c_1}}\right)^{\frac{1}{2}}-r\phi_r \left(\ln{\frac{2^{\frac{r}{\eta}}-1}{c_1}}\right)^{-\frac{1}{2}} \right)\Bigg|_{g(\theta_m)}^{g(0)}.
\end{align}\normalsize

 Depending on the magnitude of maximum alignment error $|\theta_{m}|$ with respect to the  main lobe beamwidth of user device, two scenarios are possible. We represent these scenarios  using the following two propositions:
\noindent \begin{proposition}\label{prop1}
If $|\theta_{m}|\le\frac{\phi_{ML}}{2}$, the expected link capacity is same as given by (\ref{finalExpectedCapacity2}).
\end{proposition} 
\noindent 

\begin{proof*}
\textnormal{If the maximum alignment error $|\theta_{m}|$, is less than half of  main lobe beamwidth $2.6\phi$,   Rx is bound to stay within its main lobe beamwidths, despite of the alignment error. Therefore, the average capacity  obtained using (\ref{finalExpectedCapacity}) is same as represented in (\ref{finalExpectedCapacity2}).}
\end{proof*}
\noindent \begin{proposition}
If $|\theta_{m}|>\frac{\phi_{ML}}{2}$, the expected link capacity is,
\begin{small}
\begin{align} 
\mathbb{E}[R]=&\frac{p_{\mbox{m,m}}}{2\theta_{m}}\left(\frac{2\phi_r\eta}{\ln{2}}\left(\ln{\frac{2^{\frac{r}{\eta}}}{c_1}}\right)^{\frac{1}{2}}-r\phi_r \left(\ln{\frac{2^{\frac{r}{\eta}}-1}{c_1}}\right)^{-\frac{1}{2}} \right)\Bigg|_{g(\theta_m)}^{g(0)}\nonumber\\
&+p_{m,s}\eta\log_2\left(1+\frac{G^{\phi_t}_mG^{\phi_r}_sG_0(\lambda,d)P_t}{N_0W}\right),
\label{average_misalignement_more}
\end{align}
\end{small}
\end{proposition}
\noindent
here $p_{m,m}$ is the probability of both Tx and Rx staying within their main lobes and $p_{m,s}$  represents Tx staying within its mainlobe while Rx staying outside its main lobe.

\begin{proof*}
\textnormal{If the maximum misalignment $|\theta_{m}|>1.3\phi$, we have two cases: (i)~when $|\theta|<1.3\phi$, with a probability $p_{m,m}=\left(\frac{1.3\phi_t}{\theta_{m}}\right)$, both Tx and Rx beams are pointing within the main lobe beamwidth. We represent the link capacity in this case by $\mathbb{E}[R||\theta|<1.3\phi]$ and as given by (\ref{finalExpectedCapacity2}). (ii)~if $|\theta_{m}|\geq|\theta|>1.3\phi$, then with  probability $p_{m,s}=\left(1-\frac{1.3\phi_t}{\theta_{m}}\right)$, Rx antenna points in the direction of side lobe. Since, we have assumed a constant side lob gain (recall (\ref{Eq:GaussianGain})), the capacity in this case, represented as $\mathbb{E}[R\mid |\theta|\geq 1.3\phi]$, can be evaluated using (\ref{jointEquation1}) by substituting the gain of Rx antenna by $G_s^{\phi_r}$. The link capacity can be written as  \small$\mathbb{E}[R||\theta|\geq 1.3\phi]=\eta\log_2\left(1+\frac{G^{\phi_t}_m(0)G^{\phi_r}_sG_0(\lambda,d)P_t}{N_0W}\right)$\normalsize.\newline
Considering the above two situations, the expression for average link  capacity is, \small $\mathbb{E}[R]=p_{m,m}\mathbb{E}[R||\theta|<1.3\phi]+p_{m,s}\mathbb{E}[R||\theta|\geq 1.3\phi]$\normalsize. By replacing the corresponding capacity values ($r_{m,m}^{\phi_t,\phi_r}$ and $r_{m,s}^{\phi_t, \phi_r}$), (\ref{average_misalignement_more}) can be calculated.
}
\end{proof*}
%
%
\section{Numerical Results and Discussions}\label{sec:results_discuss}
We use MATLAB to examine the effect of misalignment and searching overheads on mmWave link performance. We assume transmit power $P_t$ = 10\,mW,  distance between Tx and Rx $d$ = 5\,m,  carrier bandwidth $W$ = 2.16\,GHz and  the equal coarse sector beamwidth for both the Tx and Rx $\Omega$ = 90\degree. Using the IEEE 802.11ad specification   the transmission time \mbox{${T_p}$} of beam-training packet is assume to be 15.6$\mu$s. The allocated slot-time duration  $T_S$ takes two values of  10\,ms and 1\,s.

%
\begin{figure}[!h]
\centering
\subfigure[$T_S$=10\,ms.]{ 
\psfrag{x}[cc][cc][0.80]{{Transmit and receive beamwidths (\textdegree)}}
\psfrag{y}[cc][cc][0.80]{{Capacity (bits/slot/Hz)}}
\includegraphics[width=0.2200\textwidth]{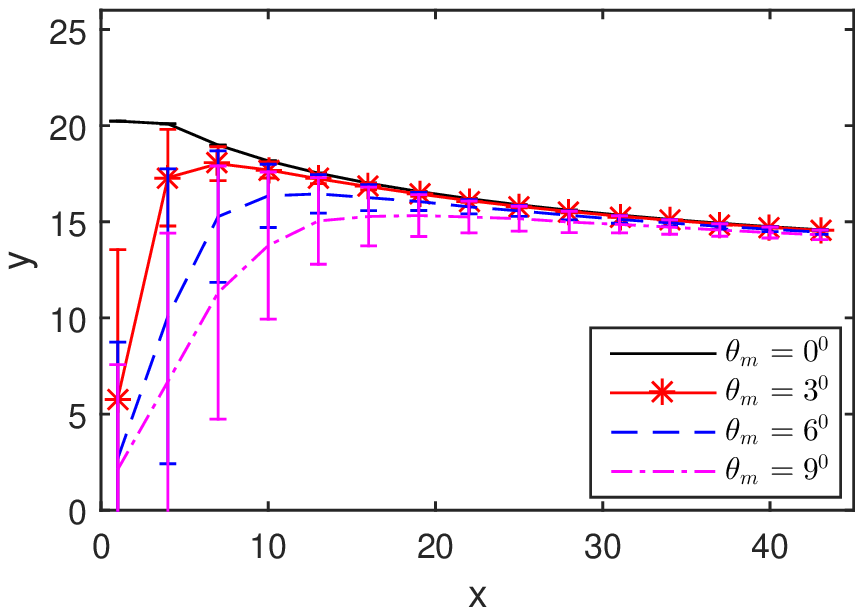}
\label{fig_Txpencilbeam_RxOmni_10ms}
}
\subfigure[$T_S$e=1\,s.]{
\psfrag{x}[cc][cc][0.80]{{Transmit and receive beamwidths (\textdegree)}}
\psfrag{y}[cc][cc][0.80]{{Capacity (bits/slot/Hz)}}
\includegraphics[width=0.2200\textwidth]{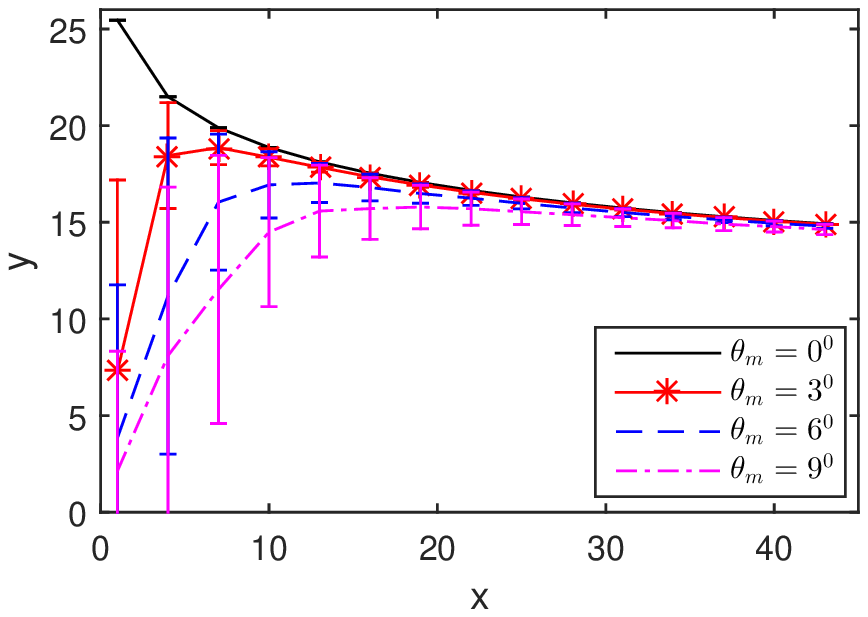}
\label{fig:Txpencilbeam_RxOmni_1sec}
}
\caption{Link capacity for pencil beam reception and coarse-sector transmission.}
\label{fig:Txpencilbeam_RxOmni}
\end{figure}
For evaluation purpose, we consider two configurations as follows: (i)~ Rx uses narrow pencil beams and Tx only employ the coarse-sector beamwidth; and (ii)~ Tx and Rx both are capable of forming narrow beams with equal beamwidth $\phi$.

Fig.~\ref{fig:Txpencilbeam_RxOmni} shows the link capacity while using pencil beam reception and the coarse-sector transmission. As it can be observed from Fig.~\ref{fig_Txpencilbeam_RxOmni_10ms}, the optimum Rx beamwidth  are 3\degree, 7\degree, 10\degree\, and 15\degree\,  for  the corresponding maximum alignment errors  $\theta_{m}$ of  0\degree, 3\degree, 6\degree\, and 9\degree, respectively.  For $\theta_{m}$ = 9\degree, using 15\degree\, Rx beamwidth outperforms the  5\degree\, beamwidth by  60\%. We can also observe that as Rx beamwidth decreases, variations in link capacity are more pronounced.

Comparing Fig.~\ref{fig_Txpencilbeam_RxOmni_10ms}  and  Fig.~\ref{fig:Txpencilbeam_RxOmni_1sec}, for the perfect alignment ($\bar{{\theta}}$ = 0\degree), we can see that  $T_S$ = 1\,s outperforms $T_S$ = 10\,ms.  Since the Tx employ coarse-sector beamwidth, the  beam searching overhead is less impact-full when $T_S$=1\,s and $\theta_{m}$ = 0\degree. Therefore, if  alignment error does not exist, then using very narrow Rx beamwidth would provide the best link capacity given that  sufficiently large slot is allocated for high speed data transmission and the coarse sector Tx beamwidth is used.

Fig.~\ref{fig:Txpencilbeam_Rxpencil_Equal} presents the link capacity when both the Tx and Rx use narrow  beamwidths. Here the  misalignment sensitivity of the mmWave links is considerably  higher.  We can see that irrespective of the slot duration $T_S$, increasing the maximum alignment error  $\theta_{m}$ always impacts the link capacity which is clearly seen by the large gaps in the maximum  and minimum link capacity for smaller beamwidths. This proves the dominance of alignment errors over the high gain achieved by employing narrow beamwidths. 
%
\begin{figure}[]
\centering
\subfigure[$T_S$=10\,ms.]{ 
\psfrag{x}[cc][cc][0.80]{{Transmit and receive beamwidths (\textdegree)}}
\psfrag{y}[cc][cc][0.80]{{Capacity (bits/slot/Hz)}}
\includegraphics[width=0.220\textwidth]{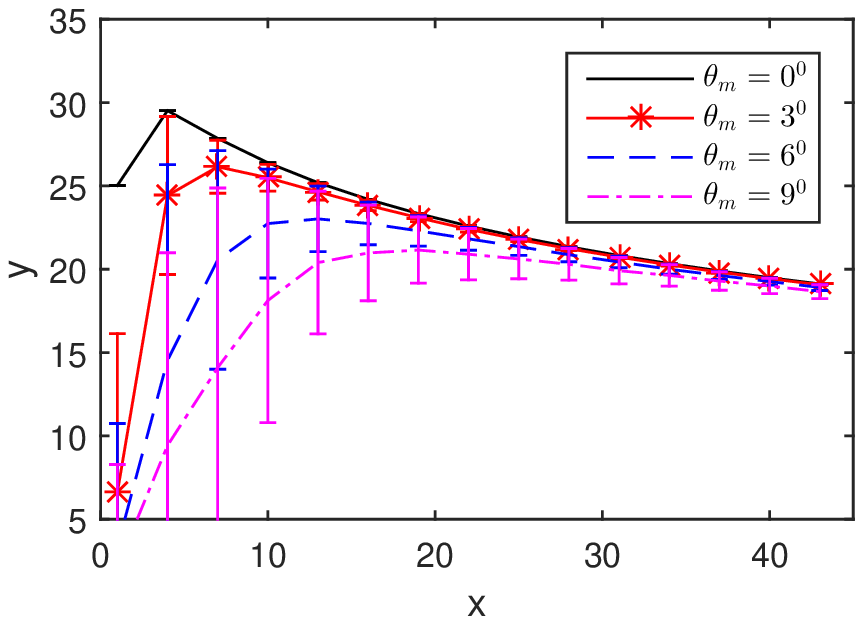}
\label{fig_Txpencilbeam_Rxpencil_Equal_10ms}
}
\subfigure[$T_S$=1\,s.]{
\psfrag{x}[cc][cc][0.80]{{Transmit and receive beamwidths (\textdegree)}}
\psfrag{y}[cc][cc][0.80]{{Capacity (bits/slot/Hz)}}
\includegraphics[width=0.220\textwidth]{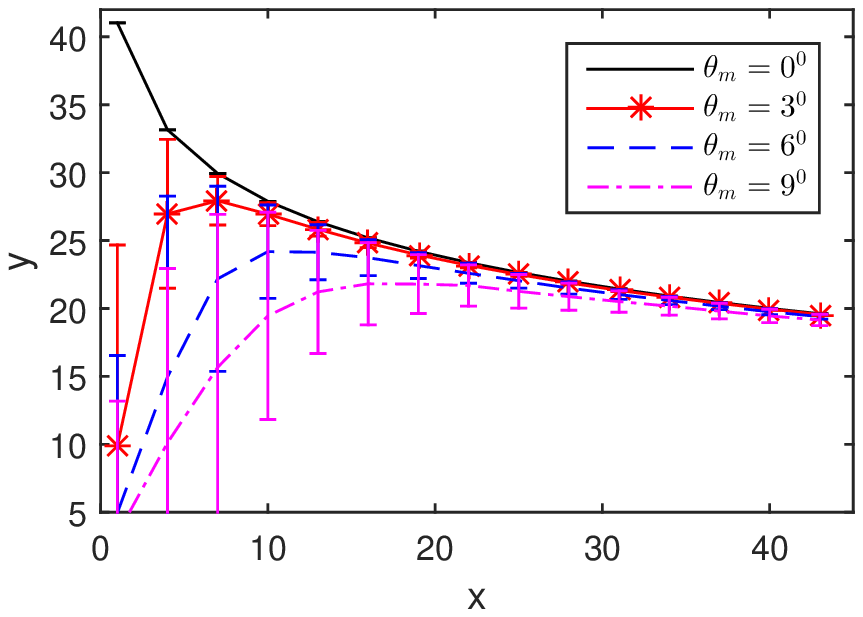}
\label{fig:Txpencilbeam_Rxpencil_Equal_1sec}
}
\caption{Link capacity for  pencil beam transmission and reception.}
\label{fig:Txpencilbeam_Rxpencil_Equal}
\end{figure}

\section{Misalignment-aware Beamwidth Adaptation Mechanism}\label{sec:Misalignment_mechanism}
 Mechanisms that can adapt link beamwidths according to the anticipated alignment dynamics are highly desired for mmWave communications for sustained link performance. We propose  a simple adaptation mechanism (see Algorithm~\ref{ALG:misalignmentaware})  based on the observed received signal strength (RSSI) where Tx and Rx beamwidths are chosen in such a way that the link can adapt to alignment errors and a sustained link performance can be achieved. The algorithm is explained as follows.\newline 
 1)~The initial Tx and Rx beamwidths $\phi_{t} $ and  $\phi_r$ are determined by Eq.~(\ref{jointEquation2}) considering perfect alignment. In practical situations,  $\phi_{t} and \phi_r$ would be the putcome of $2^{nd}$ stage of beam-training.\newline
 2)~We define a threshold RSSI ($\mbox{RSSI}_{\mbox{th}}$). If the average RSSI~($\overline{\mbox{RSSI}}$)  during a time slot falls below the threshold $\mbox{RSSI}_{\mbox{th}}$, we adapt the link beamwidth by increasing it in small-steps of $\Delta\phi$. $\mbox{RSSI}_{\mbox{th}}$ is calculated using the averaging the RSSI considering antennas pointing randomly within the entire main lobe beamwidth, i.e.,  $\theta \in [-1.3\phi, 1.3\phi$].  
 \begin{equation*}
 \mbox{RSSI}_{\mbox{th}}=\frac{1}{(\frac{2.6\phi}{\Delta{\theta}}+1)}\sum_{\theta=-1.3\phi}^{1.3\phi}G_{mt}^{\phi_t}(\theta)G_{mr}^{\phi_r}(\theta)G_{0}(\lambda, d)P_t.
 \end{equation*}
 Here ${\Delta{\theta}}$ is the sampling interval and $\frac{1}{(\frac{2.6\phi}{\Delta{\theta}}+1)}$ is the number of steps. In our evaluations, we considered ${\Delta{\theta}} =$ 2\textdegree.\newline
 3)~The beamwidth adaptation is repeated until $\overline{\mbox{RSSI}}$  reaches $\mbox{RSSI}_{\mbox{th}}$.

\begin{algorithm}
\caption{Beamwidth adaptation algorithm}
\label{ALG:misalignmentaware}
\begin{algorithmic}[1]
\STATE {Begin with the Tx/Rx beamwidths determined by (\ref{jointEquation2}) without considering alignment error for the given slot duration $T_S$};
\STATE {Monitor the average signal strength $\overline{\mbox{RSSI}}$ during transmission};
\IF {$\overline{\mbox{RSSI}}<{\mbox{RSSI}}_{\mbox{th}}$}
\STATE{$\phi_t=\phi_t+\Delta\phi$ and $\phi_r=\phi_r+\Delta\phi$};
\STATE{Go to Step 2};
\ELSE 
\STATE{Stop the beam adaptation mechanism};
\RETURN {$\phi_t, \phi_r$};
\ENDIF
\end{algorithmic}
\end{algorithm}

 To evaluate the proposed scheme we consider two Tx/Rx beamwidth configurations. The $1^{st}$ configuration uses $\phi$=2\textdegree\,  and the $2nd$ uses $\phi$=7\textdegree. The initial Tx/Rx beamwidths are deduced from (\ref{jointEquation2}) considering $\theta=0$, slot length  $T_S$=1\,s and $T_S$=10\,ms. All the other simulation parameters are as used in the previous section. 

For beamwidth $\phi$=2\textdegree\,, we consider maximum alignment errors  $\theta_{m}$=2\textdegree\, and $\theta_{m}$=10\degree. For beamwidth $\phi$=7\degree, maximum  alignment errors $\theta_{m}$=7\degree\, and $\theta_{m}$=15\degree\, are assumed. We run the adaptation mechanism for 10 consecutive slots. It can be inferred from  Fig.~\ref{fig:misalignmentawareALGO} that the  adaptation mechanism gradually attains the highest achievable  capacity. In addition, the beamwidth adaptation mechanism is less likely to let the re-beamforming be invoked. 

 We can see that for 7\degree\, beamwidth, the proposed mechanism considerably improves the link capacity by 20\% - 100\% (see Fig.~\ref{fig:mialignmentAlgo_BW15_mis30} and Fig.~\ref{fig:mialignmentAlgo_BW15_mis45}). For the 2\degree\, beamwidth, similar improvements in link capacity are registered~(see Fig.~\ref{fig:mialignmentAlgo_BW5_mis20} and Fig.~\ref{fig:mialignmentAlgo_BW5_mis20}). By comparing Fig.~\ref{fig:mialignmentAlgo_BW5_mis10} and  Fig.~\ref{fig:mialignmentAlgo_BW5_mis20}, we  see that  in Fig.~\ref{fig:mialignmentAlgo_BW5_mis10} it takes  two slots to attain the peak link capacity while in Fig.~\ref{fig:mialignmentAlgo_BW5_mis20} requires five slots to attain the peak link capacity. This is the consequence of higher  alignment errors in the latter case where beamwidth adaptation requires more slots to adjust the beamwidths according to the maximum alignment error.  It is important to note that  we have always assumed that the link deterioration is always caused by alignment errors. On the other hand, mmWave links are highly susceptible to blockage. In that case, appropriate mechanisms are required to identify if link is disrupted due to  blockage or misalignment.\normalcolor
\begin{figure}[]
\centering
\subfigure[${\phi}$=7\degree,  $\theta_{m}$=7\degree.]{
\psfrag{x}[cc][cc][0.80]{{Slot number}}
\psfrag{y}[cc][cc][0.80]{{Capacity}}
\includegraphics[width=0.22\textwidth]{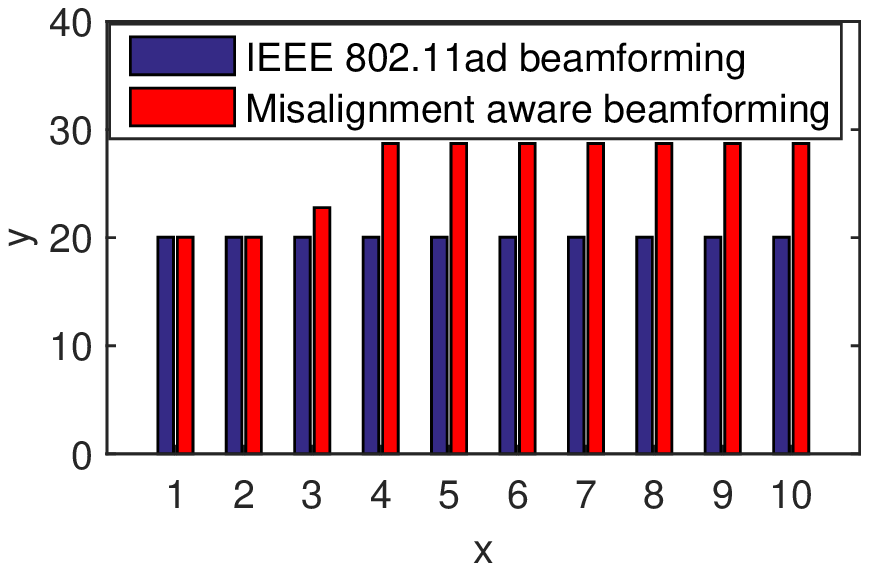}
\label{fig:mialignmentAlgo_BW15_mis30}
}
\subfigure[${\phi}$=7\degree,  $\theta_{m}$=15\degree.]{
\psfrag{x}[cc][cc][0.80]{{Slot number}}
\psfrag{y}[cc][cc][0.80]{{Capacity}}
\includegraphics[width=0.22\textwidth]{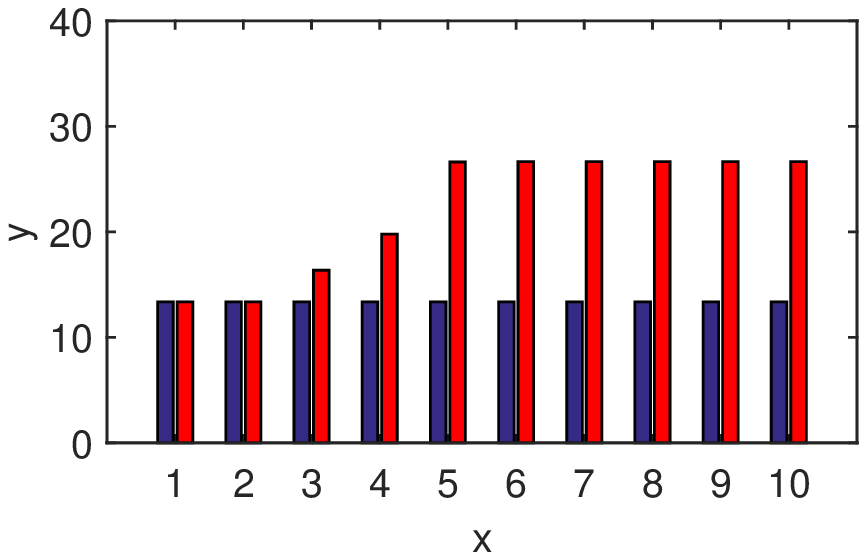}
\label{fig:mialignmentAlgo_BW15_mis45}
}
\subfigure[${\phi}$=2\degree,  $\theta_{m}$=2\degree.]{ 
\psfrag{x}[cc][cc][0.80]{{Slot number}}
\psfrag{y}[cc][cc][0.80]{{Capacity}}
\includegraphics[width=0.220\textwidth]{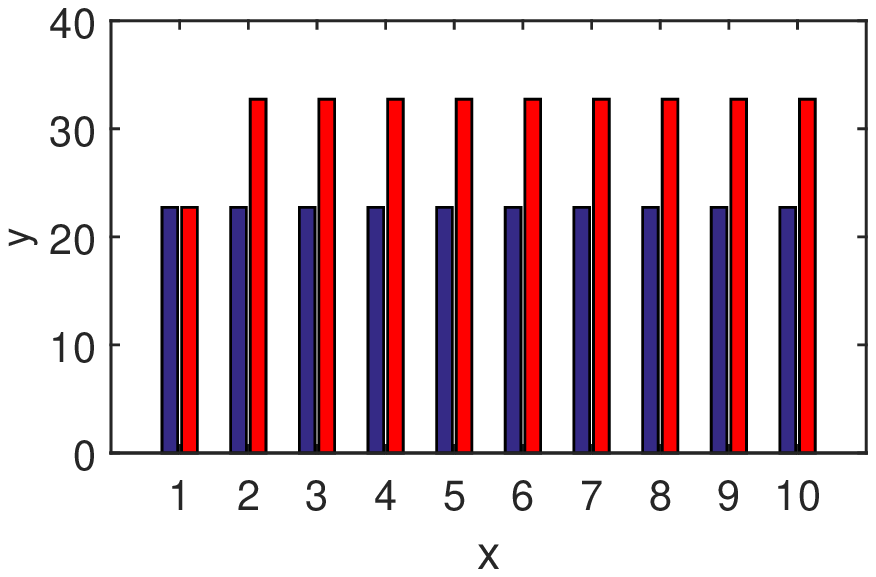}
\label{fig:mialignmentAlgo_BW5_mis10}
}
\subfigure[${\phi}$=2\degree,  $\theta_{m}$=10\degree.]{
\psfrag{x}[cc][cc][0.80]{{Slot number}}
\psfrag{y}[cc][cc][0.80]{Capacity}
\includegraphics[width=0.220\textwidth]{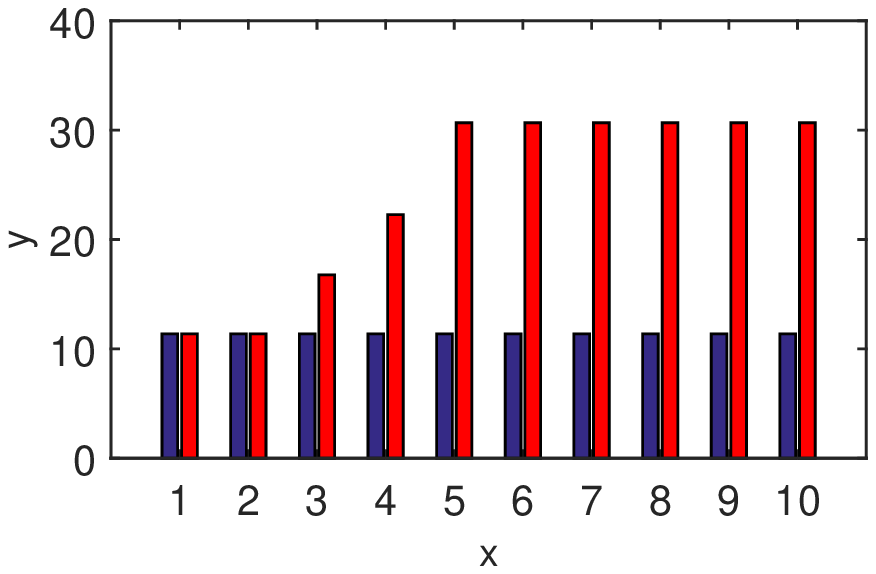}
\label{fig:mialignmentAlgo_BW5_mis20}
}
\caption{Performance of misalignment-aware beamwidth adaptation mechanism. Here capacity is in (bits/slot/Hz).}
\label{fig:misalignmentawareALGO}
\end{figure}
\section{Conclusion}\label{sec:conclusion}
\normalsize
Ultra-low latency and extreme-high reliability are two defining characteristics of Tactile Internet (TI).
Since radio interface-diversity is highly desired for attaining the carrier-grade reliability in TI, Millimeter Wave (mmWave) communication is an appropriate candidates due to its unique propagation characteristics and potential for high data rate transmissions. However, using directional antennas makes mmWave susceptible to link instability and high latency. In this paper, we proposed a capacity modeling framework for the mmWave links which provided detailed insights in to the impact of using narrow beamwidth links.  Our evaluations suggest that to exploit the multi-Gbps mmWave wireless links for reliable high data rate TI applications, beamwidth of directional antennas must be carefully decided. When alignment errors are likely to happen, using moderately narrow beams  instead of very narrow beams is beneficial. We have also proposed a  beamwidth-adaptation mechanism that is able to significantly stabilize the performance of mmWave link which is quintessential for TI applications.\newline 
\bibliographystyle{IEEEtran}
\bibliography{mmWaveSchedulingref}

\appendix
In order to prove the quasi-concavity of function ${1 - (\frac{{2\pi }}{{{\Omega _t}}} + \frac{{2\pi }}{{{\Omega _r}}})\frac{{{T_p}}}{{{T_s}}} - (\frac{{{\Omega _t}}}{{{x}}})\frac{{{T_p}}}{{{T_s}}}}$, given the theorem in \eqref{theo:strict_conacve}, we consider
\begin{align} \label{eq:prod_concave} \notag
&1 - (\frac{{2\pi }}{{{\Omega _t}}} + \frac{{2\pi }}{{{\Omega _r}}})\frac{{{T_p}}}{{{T_s}}} - (\frac{{{\Omega _t}}}{x})\frac{{{T_p}}}{{{T_s}}} > a \Rightarrow x > \frac{{{\Omega _t}\frac{{{T_p}}}{{{T_s}}}}}{{1 - (\frac{{2\pi }}{{{\Omega _t}}} + \frac{{2\pi }}{{{\Omega _r}}})\frac{{{T_p}}}{{{T_s}}} - a}}\\
&1 - (\frac{{2\pi }}{{{\Omega _t}}} + \frac{{2\pi }}{{{\Omega _r}}})\frac{{{T_p}}}{{{T_s}}} - (\frac{{{\Omega _t}}}{y})\frac{{{T_p}}}{{{T_s}}} > a \Rightarrow y > \frac{{{\Omega _t}\frac{{{T_p}}}{{{T_s}}}}}{{1 - (\frac{{2\pi }}{{{\Omega _t}}} + \frac{{2\pi }}{{{\Omega _r}}})\frac{{{T_p}}}{{{T_s}}} - a}}
\end{align}
Therefore,  
\begin{align} \notag
&\lambda x + (1 - \lambda )y > \frac{{{\Omega _t}\frac{{{T_p}}}{{{T_s}}}}}{{1 - (\frac{{2\pi }}{{{\Omega _t}}} + \frac{{2\pi }}{{{\Omega _r}}})\frac{{{T_p}}}{{{T_s}}} - a}}\\ \notag
&\Rightarrow 1 - (\frac{{2\pi }}{{{\Omega _t}}} + \frac{{2\pi }}{{{\Omega _r}}})\frac{{{T_p}}}{{{T_s}}} - (\frac{{{\Omega _t}}}{{\lambda x + (1 - \lambda )y}})\frac{{{T_p}}}{{{T_s}}} > a.
\end{align}
In addition, we need to prove that ${1 - (\frac{{2\pi }}{{{\Omega _t}}} + \frac{{2\pi }}{{{\Omega _r}}})\frac{{{T_p}}}{{{T_s}}} - (\frac{{{\Omega _t}}}{{{x}}})\frac{{{T_p}}}{{{T_s}}}}$ is also nonnegative. In order to have that, 
\begin{align} \label{eq:inequ}
x \ge \frac{{{\Omega _t}}}{{\frac{{{T_p}}}{{{T_s}}} - (\frac{{2\pi }}{{{\Omega _t}}} + \frac{{2\pi }}{{{\Omega _r}}})}}
\end{align}
should hold. Given the value of the parameters $T_p$, $T_s$, $\Omega _t$ and $\Omega _r$, the right hand side of the inequality is always negative. Therefore, \eqref{eq:inequ} holds.
\end{document}